\documentclass[journal]{IEEEtran}
\usepackage{graphicx,url}
\usepackage{dcolumn}
\usepackage{bm}
\usepackage{amsmath,amsfonts,amssymb,citesort}

\newtheorem{definition}{\noindent{\it Definition}}[section]
\newtheorem{theorem}{\noindent{\it Theorem}}[section]

\newtheorem{lemma}{\noindent{\it Lemma}}[section]

\newtheorem{remark}{\noindent{\it Remark}}[section]
\newtheorem{corollary}{\noindent{\it Corollary}}[section]
\newtheorem{example}{\noindent{\it Example}}[section]

\begin{document}

\title{On the Construction of Nonbinary Quantum BCH Codes}
\author{Giuliano G. La Guardia
\thanks{Giuliano Gadioli La Guardia is with Department of Mathematics and Statistics,
State University of Ponta Grossa -- UEPG, 84030-900, Ponta Grossa, PR, Brazil.
E-mail:~{\tt \small gguardia@uepg.br}.
}}

\maketitle

\begin{abstract}
Four quantum code constructions generating several new families of
good nonbinary quantum nonprimitive non-narrow-sense
Bose-Chaudhuri-Hocquenghem codes are presented in this paper. The
first two ones are based on Calderbank-Shor-Steane (CSS)
construction derived from two nonprimitive
Bose-Chaudhuri-Hocquenghem codes. The third one is based on
Steane's enlargement of nonbinary CSS codes applied to suitable
sub-families of nonprimitive non-narrow-sense
Bose-Chaudhuri-Hocquenghem codes. The fourth construction is
derived from suitable sub-families of Hermitian dual-containing
nonprimitive non-narrow-sense Bose-Chaudhuri-Hocquenghem codes.
These constructions generate new families of quantum codes whose
parameters are better than the ones available in the literature.
\end{abstract}


\section{Introduction}

Constructions of quantum codes with good parameters are much
investigated in the literature
\cite{Calderbank:1998,Steane:1999,Bierbrauer:2000,Ketkar:2006,Chen:2005,Cohen:1999,Grassl:1999B,Grassl:2003,Xi:2004,MaLu:2006,Salah:2007,Hamada:2008,LaGuardia:2009,LaGuardia:2010,LaGuardia:2011}.
The CSS construction, the Hermitian construction, as well as the
symplectic construction are the most utilized construction methods
in order to generate good quantum codes. In this context, many
classical codes involved in these constructions are
Bose-Chaudhuri-Hocquenghem codes
\cite{BoRCha:1960,BoRChaf:1960,Hocque:1959}. Interesting works
concerning this class of codes were presented in the literature
\cite{Yue:2000,MaLu:2006,Salah:2007,LaGuardia:2009,LaGuardia:2010,LaGuardia:2011}.
More precisely, the dimension and sufficient condition (in some
cases, necessary and sufficient condition) for dual (Euclidean and
Hermitian) containing Bose-Chaudhuri-Hocquenghem codes were
investigated.

In \cite{Salah:2007}, the authors constructed families of good
nonbinary quantum (narrow-sense) codes by showing useful
properties of cyclotomic cosets. More specifically, they computed
the exact dimension of classical narrow-sense
Bose-Chaudhuri-Hocquenghem codes of length $n$ with minimum
distance of order $\mathcal{O}(n^{1/2})$ as well as establishing
useful conditions for identifying dual-containing (Euclidean as
well as Hermitian) Bose-Chaudhuri-Hocquenghem codes. Following
this approach, the authors of \cite{MaLu:2006,Yue:2000} also have
constructed quantum Bose-Chaudhuri-Hocquenghem codes by using
properties of suitable cyclotomic cosets and also dual-containing
codes. In \cite{LaGuardia:2009,LaGuardia:2010}, new families of
nonbinary quantum Bose-Chaudhuri-Hocquenghem codes were
constructed by means of the CSS, Hermitian and also by using
Steane's code construction applied to suitable sub-families of
Bose-Chaudhuri-Hocquenghem codes. Finally, new quantum MDS codes
of non Reed-Solomon type are constructed in \cite{LaGuardia:2011}.

Motivated by the construction of new nonbinary quantum codes with
good parameters, we propose four quantum code constructions
generating new families of good codes. These new families consist
of quantum codes whose parameters are better than the ones
available in the literature. In other words, fixing $n$ and $d$,
the new quantum codes achieve greater values of the number of
encoded qudits than the codes available in the literature (see
Tables I to IV). In this paper we only consider nonprimitive
codes. In order to construct these new families it is necessary to
know the exactly dimension of the classical
Bose-Chaudhuri-Hocquenghem codes used for this purpose. This is a
difficult task since the dimension of these codes is not known. To
solve this problem, we show suitable properties of cyclotomic
cosets, providing the exact dimension and also lower bounds for
the minimum distance of the corresponding quantum codes as in the
Euclidean as well as in the Hermitian case. Additionally, by
applying the concept of linear congruence, we prove (for codes of
prime length) the existence of, at least, one $q$-ary coset
containing two consecutive integers. By means of this result we
also construct new families of good nonbinary quantum codes, since
this technique allows the construction of quantum codes with great
dimension and great minimum distance.

The proposed families have parameters

\begin{itemize}
\item ${[[n, n-4(c-2)-2, d\geq c]]}_{q}$,
\end{itemize}
where $q \geq 4$ is a prime power, $n$ is an integer such that
$\gcd (q, n) = 1$, $(q - 1) \mid n$, $m= \
{{\operatorname{ord}}_{n}}(q) = 2$ and $2\leq c \leq r$, where $r$
is such that $n = r(q-1)$;

\begin{itemize}
\item ${[[n, n-2mr, d\geq r+2]]}_{q}$,
\end{itemize}
where $m= \ {{\operatorname{ord}}_{n}}(q)\geq 2$, $n$ is a prime
number and $r$ is the number of cosets satisfying suitable
conditions (see Theorem~\ref{nonpri7});

\begin{itemize}
\item ${[[n, n-m(2r-1), d\geq r+2]]}_{q}$,
\end{itemize}
where $m= \ {{\operatorname{ord}}_{n}}(q)\geq 2$, $n$ is a prime
number and $q\geq 3$;

\begin{itemize}
\item ${[[n, n-4c, d\geq c+2]]}_{q}$,
\end{itemize}
where $n > q$ is an integer with $\gcd (q, n) = 1$, $(q - 1) \mid
n$, $m= \ {{\operatorname{ord}}_{n}}(q)=2$, $1\leq c \leq r - 3$
and $r > 3$ is such that $n = r(q-1)$;

\begin{itemize}
\item ${[[n, n-4c-2, d\geq c+2]]}_{q}$,
\end{itemize}
where $2\leq c \leq r-2$, $q > 3$, $n = r(q^2 - 1)$, $r>1$ and $m=
\ {{\operatorname{ord}}_{n}}(q^2)=2$;

\begin{itemize}
\item ${[[n, n-2mr, d\geq r+2]]}_{q}$,
\end{itemize}
where $q \geq 3$ is a prime power, $n > q^2$ is a prime number
such that $\gcd (q, n) = 1$, $m= \
{{\operatorname{ord}}_{n}}(q^2)\geq 2$ and $r$ is the number of
cosets satisfying suitable conditions (see
Theorem~\ref{nonpri7h}).

This paper is structured as follows. In Section~\ref{sec2} we
recall basic concepts on cyclic codes. In Section~\ref{sec3}, the
four new quantum code constructions are presented. More precisely:
in Subsection~\ref{secII}, new families of nonprimitive quantum
codes of length $n$, where $m= \ {{\operatorname{ord}}_{n}}(q)=2$,
are generated; in Subsection~\ref{secIII}, new families of $q$-ary
quantum nonprimitive non-narrow-sense Bose-Chaudhuri-Hocquenghem
codes of prime length, where $m= \
{{\operatorname{ord}}_{n}}(q)\geq 2$, are constructed; in
Subsection~\ref{secIV}, new families of quantum codes derived from
Steane's code construction are shown; in Subsection~\ref{secV},
the construction of new families of quantum codes derived from
nonprimitive non-narrow-sense Hermitian dual-containing
Bose-Chaudhuri-Hocquenghem codes are proposed. In
Section~\ref{sec4}, the parameters of the new quantum codes are
compared with the ones available in the literature. Finally, in
Section~\ref{sec5}, a summary of this paper is given.

\section{Review of Cyclic Codes}\label{sec2}

This section presents some basic concepts on cyclic codes, necessary for the
development of this paper. For more details, we refer the reader to
\cite{Macwilliams:1977}.

Throughout this paper, $p$ denotes a prime number, $q\neq 2$ is a
prime power, ${\mathbb F}_{q}$ is a finite field with $q$
elements, $n$ is the code length (we always consider that $\gcd
(n, q) = 1$). If $C$ is an ${[n, k, d]}_{q}$ code then $C^{\perp}$
denotes its Euclidean dual and ${C}^{{\perp}_{H}}$ denotes its
Hermitian dual. As usual, $m= \ {{\operatorname{ord}}_{n}}(q)$
denotes the multiplicative order of $q$ modulo $n$ (i.e., the
smallest positive integer $m$ such that $n$ divides $q^{m}-1$) and
${\mathbb{C}}_{[s]}$ denotes the $q$-ary cyclotomic coset modulo
$n$ containing $s$, defined by ${\mathbb{C}}_{s} = \{s, sq,
sq^{2}, sq^{3},\ldots, sq^{m_{s}-1} \}$ ($m_{s}$ is the smallest
positive integer such that $sq^{m_{s}} \equiv s$ mod $n$), where
$s$ is not necessarily the smallest number in the coset
${\mathbb{C}}_{[s]}$. The minimal polynomial over ${\mathbb
F}_{q}$ of $\beta\in {\mathbb F}_{q^m}$ is the monic polynomial of
smallest degree, $M(x)$, with coefficients in ${\mathbb F}_{q}$
such that $M(\beta)=0$. If $\beta ={\alpha}^i$ for some primitive
$n$th root of unity $\alpha \in {\mathbb F}_{q^m}$ then the
minimal polynomial of $\beta ={\alpha}^i$ is denoted by
${M}^{(i)}(x)$. It is well known that $ x^{n}- 1 =
\displaystyle\prod_{s} M^{(s)}(x)$, where $M^{(s)}(x)$ denotes the
minimal polynomial of ${\alpha}^{s} \in {\mathbb F}_{q^m}$ over
${\mathbb F}_{q}$, and $s$ runs through the coset representatives
mod $n$. Let $C$ be a cyclic code of length $n$. Then there is
only one monic polynomial $g(x)$ with minimal degree in $C$ such
that $g(x)$ is the generator polynomial of $C$, where $g(x)$ is a
factor of $x^{n}-1$. The dimension of $C$ equals $n - \deg g(x)$.
The (Euclidean) dual code ${C}^{\perp}$ of a cyclic code is cyclic
and has generator polynomial ${{g(x)}^{\perp}}= x^{\deg
h(x)}h(x^{-1})$, where $h(x)=(x^{n}-1)/g(x)$. Thus, the code
having generator polynomial $h(x)$ is equivalent to the dual code
${C}^{\perp}$.

Let ${\mathbb F}_{q}$ be a finite field and $n$ a positive integer
with $\gcd(q, n)=1$. Let $\alpha$ be a primitive $n$th root of
unity. Recall that a cyclic code of length $n$ over ${\mathbb
F}_{q}$ is a Bose-Chaudhuri-Hocquenghem (BCH) code of designed
distance $\delta$ if, for some integer $b\geq 0$ we have
$$g(x)= \operatorname{lcm} \{{M}^{(b)}(x), {M}^{(b+1)}(x), \ldots,
{M}^{(b+\delta-2)}(x)\},$$ that is, $g(x)$ is the monic polynomial
of smallest degree over ${\mathbb F}_{q}$ having ${{\alpha}^{b}},
{{\alpha}^{b+1}},\ldots, {{\alpha}^{b+\delta-2}}$ as zeros. If
$n=q^m - 1$ then the BCH code is called primitive and if $b=1$ it
is called narrow-sense.

\begin{theorem}\label{CC}\cite[pg. 201]{Macwilliams:1977}
(The BCH bound) Let $C$ be a cyclic code with generator polynomial
$g(x)$ such that, for some integers $b\geq 0,$ $\delta\geq 1,$ and
$\alpha\in {\mathbb F}_{q^m}$ ($\alpha$ is a primitive $n$th root
of unity), we have $g({\alpha}^{b}) = g({\alpha}^{b+1})= \ldots =
g({\alpha}^{b+\delta-2})=0$, that is, the code has a sequence of
$\delta-1$ consecutive powers of $\alpha$ as zeros. Then the
minimum distance of $C$ is, at least, $\delta.$
\end{theorem}
From the BCH bound, the minimum distance of a BCH code is greater
than or equal to its designed distance $\delta$.

\section{Code Constructions}\label{sec3}

In this section we present our contributions, i.e., the four
quantum code constructions previously mentioned.

\subsection{Construction I - Nonprimitive Codes}\label{secII}

In this subsection we construct new families of nonbinary CSS
codes derived from two distinct classical BCH codes, not
necessarily dual-containing. To proceed further, let us recall the
so-called CSS construction:

\begin{definition}\cite{Calderbank:1998,Steane:1999,Nielsen:2000,Ketkar:2006}
Let $C_1$ and $C_2$ denote two classical linear codes with
parameters ${[n, k_1, d_1]}_{q}$ and ${[n, k_2, d_2]}_{q}$,
respectively, such that $C_2\subset C_1$. Then there exists an
${[[n, K = k_1-k_2, d]]}_{q}$ quantum code, where $d = \min \{
wt(c) \mid c \in (C_1 \backslash C_2) \cup (\displaystyle C_{2}^{\perp}
\backslash \displaystyle C_{1}^{\perp}) \}$.
\end{definition}

We start by showing Lemma~\ref{fewco}:

\begin{lemma}\label{fewco}
Let $q\geq 3$ be a prime power and $n> q$ be an integer such that
$\gcd (q, n) = 1$. Assume also that $(q - 1) \mid n$ and $m= \
{{\operatorname{ord}}_{n}}(q)\geq 2$ hold. Then each one of the
$q$-ary cyclotomic cosets ${\mathbb{C}}_{[lr]}$, where $r$ is such
that $n = r(q-1)$ and $1 \leq l \leq q - 2$ is an integer, has
only one element.
\end{lemma}
\begin{proof}
Since $rq = n + r$ holds, one has $(lr)q = l(n + r) \equiv lr \mod
n$, and therefore $(lr)q^{t} \equiv lr \mod n$, for each $ 1 \leq
t \leq m-1$, proving the lemma.
\end{proof}

Lemma~\ref{fewco} can be applied in order to show
Theorem~\ref{nonpri3}.

\begin{theorem}\label{nonpri3}
Assume that $q > 3$ is a prime power and $n > q$ is an integer
such that $\gcd (q, n) = 1$. Assume also that $(q - 1) \mid n$ and
$m= \ {{\operatorname{ord}}_{n}}(q)=2$ hold. Then there exist
quantum codes with parameters ${[[n, n-4(r-2)-2, d\geq r]]}_{q}$,
where $r$ is such that $n = r(q-1)$.
\end{theorem}
\begin{proof}
Since it is true that $n \mid (q^{2} - 1)$ and because we consider
only nonprimitive BCH codes, it follows that $r\leq q$. Since
$\gcd (q, n) = 1$ one has $r < q$, so the inequalities $(r-2)q <
n$ and $r + (r-2)q < n$ hold. We next show that all the $q$-ary
cosets (modulo $n$) given by ${\mathbb{C}}_{[0]}=\{ 0 \},
{\mathbb{C}}_{[1]}= \{ 1, \ \ q \}, {\mathbb{C}}_{[2]}= \{ 2, \ \
2q \}, {\mathbb{C}}_{[3]}= \{ 3, \ \ 3q \}, \ldots,
{\mathbb{C}}_{[r-2]}= \{ r-2, \ \ (r-2)q \}, {\mathbb{C}}_{[r]}=
\{ r\}, {\mathbb{C}}_{[r+1]}= \{ r+1, \ \ r + q\},
{\mathbb{C}}_{[r+2]}= \{ r+2, \ \ r + 2q\}, \ldots,
{\mathbb{C}}_{[2r-2]}= \{ 2r-2, \ \ r + (r-2)q \}$, are mutually
disjoint and, with exception of the cosets ${\mathbb{C}}_{[0]}=\{
0 \}$ and ${\mathbb{C}}_{[r]} = \{r\}$, each of them has exactly
two elements.

The cosets ${\mathbb{C}}_{[0]}$ and ${\mathbb{C}}_{[r]}$ have only one element.
Let us show that each one of the other cosets has exactly two elements. Since
$(r-2)q < n$, then the congruence $l \equiv lq$ mod $n$ implies that $l = lq$, where
$1 \leq l\leq r - 2$, which is a contradiction. If $r + s \equiv (r + s)q$ mod $n$, where
$1 \leq s\leq r - 2$, then $r + s = r + sq$, which is a contradiction.

From now on, we show that all these cosets given above and
${\mathbb{C}}_{[0]}$ and ${\mathbb{C}}_{[r]}$ are mutually
disjoint. We only consider the case
${\mathbb{C}}_{[r+l]}={\mathbb{C}}_{[r-s]}$, where $1\leq l, s\leq
r-2$, since the other cases are similar to this one. Seeking a
contradiction, we assume that
${\mathbb{C}}_{[r+l]}={\mathbb{C}}_{[r-s]}$, where $1\leq l, s\leq
r-2$. If the congruence $(r + l) \equiv (r-s)$ mod $n$ holds, one
obtains
\begin{eqnarray*}
(r + l) \equiv (r-s) \ mod \ n \Longrightarrow n \mid (l+s).
\end{eqnarray*}
If $l+s \neq 0$ one has $n\leq l+s$, which is a contradiction. If
$l+s = 0$ holds it implies that $l=-s$, which is a contradiction.

On the other hand, if $(r + l)q \equiv r-s$ mod $n$ holds, one obtains
\begin{eqnarray*}
(r + l)q \equiv r-s \Longrightarrow lq \equiv -s \ mod \ n\\
\Longrightarrow n \mid (lq+s).
\end{eqnarray*}
Since $l, s \leq r-2$ and $r < q$ hold, if $lq+s \neq 0$ holds it
follows that $lq +s < n$, which is a contradiction. If $lq+s = 0$
then $lq =-s$, which is a contradiction. Thus all the $q$-ary
cosets ${\mathbb{C}}_{[0]}, \ {\mathbb{C}}_{[1]}, \ldots,
{\mathbb{C}}_{[r-2]}$, are disjoint from each one of the $q$-ary
cosets ${\mathbb{C}}_{[r]}, \ {\mathbb{C}}_{[r+1]}, \ldots,
{\mathbb{C}}_{[2r-2]}$. Additionally, all the $q$-ary cosets
${\mathbb{C}}_{[0]}, \ {\mathbb{C}}_{[1]}, \ldots,
{\mathbb{C}}_{[r-2]}$, are mutually disjoint and all the $q$-ary
cosets ${\mathbb{C}}_{[r]}, \ {\mathbb{C}}_{[r+1]}, \ldots,
{\mathbb{C}}_{[2r-2]}$, are also mutually disjoint.

Let $C_1$ be the cyclic code generated by the product of the
minimal polynomials
\begin{eqnarray*}
M^{(0)}(x)M^{(1)}(x)\cdot\ldots\cdot M^{(r-2)}(x),
\end{eqnarray*}
and $C_2$ be the cyclic code generated by $g_2(x)$, that is the
product of the minimal polynomials
\begin{eqnarray*}
g_2(x) = \displaystyle\prod_{i} M^{(i)}(x),
\end{eqnarray*} where
$i \notin \{ r, r+1, \ldots, 2r-2 \}$ and $i$ runs through the
coset representatives mod $n$. From construction one has $C_2
\subsetneq C_1$. From the BCH bound, the minimum distance of $C_1$
is greater than or equal to $r$ because its defining set contains
the sequence $0, 1, \ldots, r-2$, of $r-1$ consecutive integers.
Similarly, the defining set of the code $C$ generated by the
polynomial $h(x)= \frac{x^n - 1}{g_2(x)}$ contains the sequence
$r, r+1, \ldots, 2r-2$, of $r-1$ consecutive integers and so, from
the BCH bound, $C$ also has minimum distance greater than or equal
to $r$. Since the code $\displaystyle C_{2}^{\perp}$ is equivalent
to $C$, $\displaystyle C_{2}^{\perp}$ also has minimum distance
greater than or equal to $r$. Therefore, the resulting CSS code
has minimum distance greater than or equal to $r$.

Next we compute the dimension of the corresponding CSS code. We
know that the degree of the generator polynomial of a cyclic code
equals the cardinality of its defining set. Further, the defining
set $Z_1$ of $C_1$ has $r - 1$ disjoint cyclotomic cosets.
Moreover, all of them (except coset ${\mathbb{C}}_{0}$) have two
elements and so, $Z_1$ has $2(r-2)+1$ elements. Therefore, $C_1$
has dimension $k_1 = n - 2(r-2)-1$. Similarly, $C_2$ has dimension
$k_2 = 2(r-2)+1$. Thus the dimension of the corresponding CSS code
equals $n - 4(r-2)-2$. Applying the CSS construction to the codes
$C_1$ and $C_2$, one can get quantum codes with parameters ${[[n,
n - 4(r-2)-2, d \geq r]]}_{q}$.
\end{proof}

We illustrate Theorem~\ref{nonpri3} by means of a graphical scheme:

\begin{eqnarray*}
\underbrace{ \overbrace{ {\mathbb{C}}_{[0]} {\mathbb{C}}_{[1]} \
{\mathbb{C}}_{[2]} \ \ldots \ {\mathbb{C}}_{[r-2]}}^{C_1}}_{C_2} \\
\overbrace{{\mathbb{C}}_{[r]} \ {\mathbb{C}}_{[r+1]} \ldots \
{\mathbb{C}}_{[2r-2]}}^{C}
\ \underbrace{{\mathbb{C}}_{[a_1]} \ldots {\mathbb{C}}_{[a_n]}}_{C_2}.\\
\end{eqnarray*}

The union of the cosets ${\mathbb{C}}_{[0]}, {\mathbb{C}}_{[1]},
\ldots, {\mathbb{C}}_{[r-2]}$ is the defining set of code $C_1$;
the union of the cosets ${\mathbb{C}}_{[0]}, {\mathbb{C}}_{[1]},
\ldots, {\mathbb{C}}_{[r-2]}, {\mathbb{C}}_{[a_1]}, \ldots,
{\mathbb{C}}_{[a_n]}$ is the defining set of $C_2$, where
${\mathbb{C}}_{[a_1]}, \ldots, {\mathbb{C}}_{[a_n]}$ are the
remaining cosets in order to complete the set of all cyclotomic
cosets. The union of the cosets ${\mathbb{C}}_{[r]},
{\mathbb{C}}_{[r+1]}, \ldots, {\mathbb{C}}_{[2r - 2]}$ is the
defining set of $C$.

\begin{corollary}\label{nonpri3cor}
Assume that all the hypothesis of Theorem~\ref{nonpri3} are valid. Then
there exist quantum codes with parameters ${[[n, n-4(c-2)-2, d\geq c]]}_{q}$,
where $2\leq c < r$.
\end{corollary}
\begin{proof}
Choose $C_1$ be the cyclic code generated by the product of the minimal polynomials
\begin{eqnarray*}
M^{(0)}(x)M^{(1)}(x)\cdot\ldots\cdot M^{(c-3)}(x)M^{(c-2)}(x),
\end{eqnarray*}
and $C_2$ be the cyclic code generated by the product of the minimal polynomials
\begin{eqnarray*}
\displaystyle\prod_{i} M^{(i)}(x),
\end{eqnarray*} where
$i \notin \{ r, r+1, \ldots, r + c-2 \}$ and $i$ runs through the
coset representatives mod $n$. Proceeding similarly as in the
proof of Theorem~\ref{nonpri3}, the result follows.
\end{proof}


\subsection{Construction II - Codes of Prime Length}\label{secIII}

In this subsection the attention is focused on cyclic codes of
prime length. Among the contributions shown in this section, we
prove there exists at least one $q$-ary cyclotomic coset
containing two consecutive integers (see Lemma~\ref{nonpril1}). In
order to proceed further, let us recall a well-known result from
number theory:

\begin{theorem}\label{folklore}
A linear congruence $ax \equiv b$ (mod $m$), where $a \neq 0$,
admits an integer solution if and only if $d = \gcd (a, m)$
divides $b$.
\end{theorem}

Applying Theorem~\ref{folklore} we can prove Lemma~\ref{nonpril1}:

\begin{lemma}\label{nonpril1}
Assume that $q \geq 3$ is a prime power, $n > q$ is a prime number
and consider $m= \ {{\operatorname{ord}}_{n}}(q)\geq 2$. Then
there exists at least one $q$-ary cyclotomic coset containing two
consecutive integers.
\end{lemma}

\begin{proof}
First, note that $\gcd (q, n) = 1$. In order to prove this lemma,
it suffices to show that the congruence $xq \equiv x + 1 (\mod n)$
has at least one solution for some $0 \leq x\leq n-1$ or,
equivalently, the congruence $(q - 1)x \equiv 1$ (mod $n$) has at
least one solution. We know that $\gcd (q - 1, n) = 1$ holds,
because $n > q$ and $n$ is a prime number. Since $q - 1 \neq 0$,
it follows from Theorem~\ref{folklore} that $(q - 1)x \equiv 1$
(mod $n$) has an integer solution $x_0$. Applying the division
algorithm for $x_0$ and $n$ one has $x_0 = n s_0  + r_0$, where
$r_0$ and $s_0$ are integers and $0\leq r_0 \leq n-1$. Since
$(q-1)x_0\equiv 1$ (mod $n$) holds then the congruence
$(q-1)r_0\equiv 1$ (mod $n$) also holds, and the result follows.
\end{proof}

\begin{remark} Note that in Lemma~\ref{nonpril1} it is not
necessary to assume that $n$ is a prime number. In fact, we only
need to suppose that $\gcd (q - 1, n) = 1$ and $\gcd (q, n) = 1$
hold (the latter condition ensures that $C$ has simple roots). But
since the corresponding $q$-ary cosets of BCH codes of prime
length have nice properties, we have assumed that $n$ is prime.
However, if one assumes that $\gcd (q - 1, n) = 1$ and $\gcd (q,
n) = 1$ hold, more good quantum codes can be constructed.
\end{remark}

\begin{theorem}\label{nonpri5}
Let $q \geq 3$ be a prime power, $n > q$ be a prime number and
consider $m= \ {{\operatorname{ord}}_{n}}(q)\geq 2$. Assume that
${\mathbb{C}}_{[s]} \neq {\mathbb{C}}_{[-s]}$, where
${\mathbb{C}}_{[s]}$ is a cyclotomic coset containing two
consecutive integers. Then there exist quantum codes with
parameters ${[[n, n-2m, d\geq 3]]}_{q}$.
\end{theorem}

\begin{proof}
First, note that $\gcd (q, n) = 1$. Choose $C_1$ be code generated
by $M^{(s)}(x)$ and $C_2$ be the code generated by
$\displaystyle\prod_{i} M^{(i)}(x)$, where $i \neq -s$ and $i$
runs through the coset representatives mod $n$. It is easy to see
that the cosets ${\mathbb{C}}_{[s]}$ and ${\mathbb{C}}_{[-s]}$
contain $m$ elements. Proceeding similarly as in the proof of
Theorem~\ref{nonpri3}, the result follows.
\end{proof}

\begin{theorem}\label{nonpri7}
Assume that $q \geq 3$ is a prime power, $n > q$ is a prime number
and consider $m= \ {{\operatorname{ord}}_{n}}(q)\geq 2$. Let
${\mathbb{C}}_{[s]}$ be the cyclotomic coset containing $s$ and
$s+1$. Suppose that all the $q$-ary cosets ${\mathbb{C}}_{[s]},
{\mathbb{C}}_{[s+2]}, \ldots, {\mathbb{C}}_{[s+r]},
{\mathbb{C}}_{[-s]}, {\mathbb{C}}_{[-s-2]}, \ldots,
{\mathbb{C}}_{[-s-r]}$, are mutually disjoint. Then there exist
quantum codes with parameters ${[[n, n-2mr, d\geq r+2]]}_{q}$.
\end{theorem}

\begin{proof}
We know that $\gcd (q, n) = 1$ and the coset ${\mathbb{C}}_{[-s]}$
also contains two consecutive integers, namely, $-s-1$ and $-s$.
Let $C_1$ be the cyclic code generated by the product of the
minimal polynomials
\begin{eqnarray*}
M^{(s)}(x)M^{(s+2)}(x)\cdot\ldots\cdot M^{(s+r)}(x),
\end{eqnarray*}
and let $C_2$ be the cyclic code generated by the polynomial
$g_2(x)$, that is the product of the minimal polynomials
\begin{eqnarray*}
g_2(x) = \displaystyle\prod_{j} M^{(j)}(x),
\end{eqnarray*} where
$j \notin \{ -s-r, \ldots, -s-2,  -s \}$ and $j$ runs through the
coset representatives mod $n$.

From the BCH bound, the minimum distance of $C_1$ is greater than
or equal to $r+2$ because its defining set contains the sequence
of $r+1$ consecutive integers given by $s, s+1, s+2, \ldots, s+r$.
Similarly, the defining set of the code $C$ generated by the
polynomial $h_2(x)=(x^{n}-1)/ g_2(x)$, contains a sequence of
$r+1$ consecutive integers given by $-s-r, \ldots, -s-2, -s-1,
-s$. Again, from the BCH bound, $C$ has minimum distance greater
than or equal to $r+2$. Since $C$ is equivalent to $\displaystyle
C_{2}^{\perp}$, it follows that $\displaystyle C_{2}^{\perp}$ also
has minimum distance greater than or equal to $r+2$. Therefore,
the resulting CSS code have minimum distance greater than or equal
to $r+2$. If $s \in [1, n-1]$ satisfies $\gcd (s, n) = 1$ then the
coset ${\mathbb{C}}_{s}$ has cardinality $m$. In fact, if
$|{\mathbb{C}}_{s}|=c < m$ it follows that $n|s(q^{c}-1)$, so
$n|(q^{c}-1)$, a contradiction. Thus, since $n$ is prime, each one
of the cosets ${\mathbb{C}}_{s}$, where $s \in [1, n-1]$, has
cardinality $m$. Additionally, from the hypothesis, all the
$q$-ary cosets ${\mathbb{C}}_{[s]}, {\mathbb{C}}_{[s+2]}, \ldots,
{\mathbb{C}}_{[s+r]}$, are mutually disjoint. Thus $C_1$ has
dimension $k_1 = n - mr$ and $C_2$ has dimension $k_2 = mr$, since
there exist $r$ disjoint $q$-ary cosets not contained in the
defining set of $C_2$, where each of them has cardinality $m$.
Therefore, the dimension $K$ of the corresponding CSS code equals
$K = n - 2mr$. Since the cosets ${\mathbb{C}}_{[s]},
{\mathbb{C}}_{[s+2]}, \ldots, {\mathbb{C}}_{[s+r]}$,
${\mathbb{C}}_{[-s]}, {\mathbb{C}}_{[-s-2]}, \ldots,
{\mathbb{C}}_{[-s-r]}$, are mutually disjoint, it follows that
$C_2 \subsetneq C_1$. Applying the CSS construction to $C_1$ and
$C_2$, one obtains an ${[[n, n-2mr, d\geq r+2]]}_{q}$ code.
\end{proof}

\begin{example}
Theorem~\ref{nonpri7} has variants as follows: to construct an
${[[19, 13, d\geq 3]]}_{7}$ code, consider $q = 7$, $n = 19$ and
$m=3$. The cosets are given by ${\mathbb{C}}_{2}= \{ 2, 14, 3\}$
and ${\mathbb{C}}_{16}= \{5, 16, 17\}$. Let $C_1$ be the cyclic
code generated by the minimal polynomial $
C_1=\langle{g_{1}}(x)\rangle=\langle M^{(2)}(x)\rangle$ and $C_2$
generated by $g_{2}(x) = \displaystyle\prod_{i} M^{(i)}(x)$, where
$i \notin \{ 16\}$ and $i$ runs through the coset representatives
mod $19$. Then an ${[[19, 13, d\geq 3]]}_{7}$ quantum code can be
constructed. Proceeding similarly, one can get quantum codes with
parameters ${[[31, 25, d\geq 3]]}_{5}$, ${[[71, 61, d\geq
3]]}_{5}$, ${[[11, 1, d\geq 4]]}_{3}$, ${[[31, 19, d\geq
4]]}_{5}$, ${[[31, 13, d\geq 5]]}_{5}$, ${[[71, 51, d\geq
4]]}_{5}$, ${[[71, 41, d\geq 6]]}_{5}$.
\end{example}

%
%
%


\subsection{Construction III - Codes Derived from Steane's Construction}\label{secIV}

In this subsection we construct new families of quantum BCH codes
of prime length by applying Steane's enlargement of nonbinary CSS
construction \cite[Corollary 4]{Hamada:2008}. These new families
have parameters better than the parameters of the quantum BCH
codes available in the literature. Let us recall Steane's code
construction:

\begin{corollary}\cite[Corollary 4]{Hamada:2008}\label{Hama}
Assume we have an $[N_0, K_0]$ linear code $L$ which contains its
Euclidean dual, $L^{\perp}\leq L$, and which can be enlarged to an
$[N_0, {K}^{'}_0]$ linear code $L^{'}$, where ${K}^{'}_0 \geq K_0
+ 2$. Then there exists a quantum code with parameters $[[N_0, K_0
+ {K}^{'}_0 - N_0, d \geq \ min \{d, \lceil \frac{q+1}{q}d^{'}
\rceil \}]]$, where $d = w (L \backslash {L^{'}}^{\perp})$ and
$d^{'} = w (L^{'} \backslash  {L^{'}}^{\perp})$.
\end{corollary}

Euclidean dual-containing cyclic codes can be derived from
Lemma~\ref{AAA}:

\begin{lemma}\cite[Lemma 1]{Salah:2007}\label{AAA}
Assume that $\gcd(q, n)=1$. A cyclic code of length $n$ over
${\mathbb F}_{q}$ with defining set $Z$ contains its Euclidean
dual code if and only if $Z\cap Z^{-1} =\emptyset$, where
$Z^{-1}=\{-z \ mod \ n\mid z \in Z\}$.
\end{lemma}

In Lemma~\ref{nonpril1} of Section~\ref{secIII} we have shown the
existence of, at least, one $q$-ary cyclotomic coset containing
two consecutive integers provided the code length is a prime
number. In what follows we show how to construct good quantum
codes of prime length by applying Steane's code construction. We
begin by presenting an illustrative example:

\begin{example}
Assume that $n=31$ and $q=5$. From Lemma~\ref{nonpril1}, there
exists a cyclotomic coset containing at least two consecutive
integers; here it is the coset ${\mathbb{C}}_{8} =\{8, 9, 14\}$.
Let $C$ be the cyclic code generated by the product of the minimal
polynomials $C=\langle g(x)\rangle=\langle
M^{(4)}(x)M^{(8)}(x)\rangle$. $C$ has defining set
$Z={\mathbb{C}}_{4}\cup {\mathbb{C}}_{8} =\{4, 7, 8, 9, 14, 20\}$
and has parameters ${[31, 25, d\geq 4]}_{5}$. From
Lemma~\ref{AAA}, it is easy to check that $C$ is Euclidean
dual-containing. Furthermore, $C$ can be enlarged to a code
$C^{'}$ with parameters ${[31, 28, d\geq 3]}_{5}$, whose generator
polynomial is $M^{(8)}(x)$. Applying Corollary~\ref{Hama} to $C$
and $C^{'}$ one obtains an ${[[31, 22, d \geq 4]]}_{5}$ code.
\end{example}

\begin{theorem}\label{nonpri5A}
Let $q \geq 3$ be a prime power, $n > q$ be a prime number and
consider that $m= \ {{\operatorname{ord}}_{n}}(q)\geq 2$. Let
${\mathbb{C}}_{[s]}$ be the $q$-ary coset containing $s$ and $s+1$
and consider that $Z={\mathbb{C}}_{[s]} \cup
{\mathbb{C}}_{[s+2]}$, where
${\mathbb{C}}_{s}\neq{\mathbb{C}}_{[s+2]}$. Assume also that
$Z\cap Z^{-1} =\emptyset$ holds. Then there exist quantum codes
with parameters ${[[n, n-3m, d\geq 4]]}_{q}$.
\end{theorem}
\begin{proof}
We know that $\gcd (q, n) = 1$. Let $C$ be the cyclic code
generated the product of the minimal polynomials $\langle
M^{(s)}(x)M^{(s+2)}(x)\rangle$. By hypothesis and from
Lemma~\ref{AAA}, we know that $C$ is Euclidean dual-containing.
$C$ has parameters ${[n, n-2m, d\geq 4]}_{q}$. Let $C^{'}$ be the
cyclic code generated by the minimal polynomial $M^{(s)}(x)$. We
know that $C^{'}$ is an enlargement of $C$ and has parameters
${[n, n-m, d\geq 3]}_{q}$. Since $m\geq 2$, then $k^{'} - k = m
\geq 2$, where $k^{'}$ denotes the dimension of $C^{'}$ and $k$
denotes the dimension of $C$. Applying Steane's code construction
to $C$ and $C^{'}$, since $\frac{q+1}{q}>1$ holds one obtains an
${[[n, n-3m, d\geq 4]]}_{q}$ code.
\end{proof}

Theorem~\ref{nonpri5A} can be generalized in the following way:

\begin{theorem}\label{nonpri7A}
Assume that $q \geq 3$ is a prime power, $n > q$ is a prime number
and consider that $m= \ {{\operatorname{ord}}_{n}}(q)\geq 2$. Let
${\mathbb{C}}_{[s]}$ be the cyclotomic coset containing $s$ and
$s+1$. Assume that $Z={\mathbb{C}}_{[s]}\cup
{\mathbb{C}}_{[s+2]}\cup \ldots\cup {\mathbb{C}}_{[s+r]}$, where
all the $q$-ary cosets ${\mathbb{C}}_{[s+i]}$, $i =0, 2, 3,\ldots,
r$, are mutually disjoint, and suppose that $Z\cap
Z^{-1}=\emptyset$. Then there exist quantum codes with parameters
${[[n, n-m(2r-1), d\geq r+2]]}_{q}$.
\end{theorem}
\begin{proof}
We know that $\gcd (q, n) = 1$. Let $C$ be the cyclic code generated by
the product of the minimal polynomials
\begin{eqnarray*}
M^{(s)}(x)M^{(s+2)}(x)\cdot\ldots \cdot M^{(s+r)}(x).
\end{eqnarray*}
Since $Z\cap Z^{-1}=\emptyset$ holds, it implies  from
Lemma~\ref{AAA} that $C$ is Euclidean dual-containing. From the
hypothesis, all the $q$-ary cosets ${\mathbb{C}}_{[s]},
{\mathbb{C}}_{[s+2]}, \ldots, {\mathbb{C}}_{[s+r]}$ are mutually
disjoint, so $C$ has dimension $k = n - mr$ and minimum distance
$d\geq r+2$. Thus $C$ has parameters ${[n, n-mr, d\geq r+2]}_{q}$.
Let $C^{'}$ be the cyclic code generated by the product of the
minimal polynomials
\begin{eqnarray*}
M^{(s)}(x)M^{(s+2)}(x)\cdot\ldots \cdot M^{(s+r-1)}(x).
\end{eqnarray*}
We know that $C^{'}$ is an enlargement of $C$ and has parameters
${[n, n-m(r-1), d\geq r+1]}_{q}$. Since $m\geq 2$ then $k^{'} - k
= m \geq 2$, where $k^{'}$ denotes the dimension of $C^{'}$ and
$k$ denotes the dimension of $C$. Applying Steane's code
construction to the codes $C$ and $C^{'}$ one obtains an ${[[n,
n-m(2r-1), d\geq r+2]]}_{q}$ code, as required.
\end{proof}

\begin{example}
In this example we construct an ${[[31, 16, d\geq 5]]}_{5}$
quantum code. For this purpose we take $n=31$ and $q=5$; then $m=
\ {{\operatorname{ord}}_{n}}(q)=3$. Let $C$ be the cyclic code
generated by the product of the minimal polynomials
$M^{(4)}(x)M^{(6)}(x)M^{(8)}(x)$. It is easy to see that $C$ is
Euclidean dual-containing and has parameters ${[31, 22, d\geq
5]}_{5}$. Let $C^{'}$ be the cyclic code generated by the product
of the minimal polynomials $M^{(4)}(x)M^{(8)}(x)$; $C^{'}$ has
parameters ${[31, 25, d\geq 4]}_{5}$.Thus there exists an ${[[31,
16, d\geq 5]]}_{5}$ quantum code.
\end{example}

We next establish Theorem~\ref{nonpri3Ste}, an analogous to
Theorem~\ref{nonpri3}.

\begin{theorem}\label{nonpri3Ste}
Suppose that $q \geq 5$ is a prime power and $n > q$ is an integer
such that $\gcd (q, n) = 1$. Assume also that $(q - 1) \mid n$ and
$m= \ {{\operatorname{ord}}_{n}}(q)=2$ hold. Then there exist
quantum codes with parameters ${[[n, n-4c, d\geq c+2]]}_{q}$,
where $1\leq c \leq r - 3$ and $r > 3$ is such that $n = r(q-1)$.
\end{theorem}
\begin{proof}
We only prove the existence of an ${[[n, n-4(r-3), d\geq
r-1]]}_{q}$ code, since the constructions of the other codes are
quite similar.

Let $C$ be the cyclic code generated by the product of the minimal polynomials
\begin{eqnarray*}
M^{(r)}(x)M^{(r+1)}(x)\cdot\ldots\cdot M^{(2r-3)}(x).
\end{eqnarray*}
From Lemma~\ref{fewco} and from the proof of
Theorem~\ref{nonpri3}, we know that the $q$-ary cosets given by $
{\mathbb{C}}_{[r]}= \{ r\}, {\mathbb{C}}_{[r+1]}= \{ r+1, \ \ r +
q\}, {\mathbb{C}}_{[r+2]}= \{ r+2, \ \ r + 2q\}, \ldots ,
{\mathbb{C}}_{[2r-3]}= \{ 2r-3, \ \ r + (r-3)q \}$ are mutually
disjoint and each of them has two elements. Therefore, $C$ has
dimension $k = n- 2(r-3) -1$ and minimum distance $d\geq r-1$.

Let us prove that $C$ is Euclidean dual-containing. In fact, if
$(r+i) \equiv -(r+j)$ mod $n$, where $0 \leq i, j\leq r-3$, it
follows that $2r + i + j \equiv 0$ mod $n$. Since the inequality
$2r + i + j < n$ holds because $q\geq 5$, one has a contradiction.
On the other hand, if $(r+i)q \equiv -(r+j)$ mod $n$ holds then
\begin{eqnarray*}
(iq + j)(q-1)\equiv 0 \mod n\Longrightarrow\\
i(q^2 - q) + j(q - 1) \equiv 0 \mod n\Longrightarrow\\
j(q - 1) \equiv i(q-1) \mod n,
\end{eqnarray*}
where the latter congruence holds because
${{\operatorname{ord}}_{n}}(q)=2$. Then the unique solution is
when $i=j$. Let us investigate this case. Seeking a contradiction,
we assume that the congruence $(r+i)q \equiv -(r+i)$ mod $n$ is
true. Then one obtains
\begin{eqnarray*}
(r+i)q \equiv -(r+i) \mod n\Longrightarrow\\
2r +i(q+1) \equiv 0 \mod n \Longrightarrow\\
r(q-3)\equiv i(q+1) \mod n.
\end{eqnarray*}
If $0\leq i\leq r - 4$, then
\begin{eqnarray*}
r(q - 3) -i(q+1) \geq \\
r(q - 3) -(r-4)(q+1)=\\
4q -4r +4 > 0,
\end{eqnarray*}
where the latter inequality holds because $r < q$ since we only
consider nonprimitive BCH codes. Moreover, the inequality $r(q -
3) -i(q+1)< n$ also holds, which is a contradiction. If $i = r -
3$ then the congruence $r(q-3)\equiv (r - 3)(q+1)$ mod $n$ holds,
that is, $4r\equiv 3(q+1)$ mod $n$ holds. Since $r\mid (q+1)$ and
$q+1 > r$ hold, it implies that $q+1\geq 2r$ so, $3(q+1) -4r \geq
2r >0$. Moreover, the inequality $3(q+1) -4r < n$ holds, which is
a contradiction. Therefore, $C$ is Euclidean dual-containing.

Let $C^{'}$ be the cyclic code generated by the product of the
minimal polynomials
\begin{eqnarray*}
M^{(r)}(x)M^{(r+1)}(x)\cdot\ldots\cdot M^{(2r-4)}(x).
\end{eqnarray*}
$C^{'}$ is an enlargement of $C$; $C^{'}$ has dimension $k^{'}= n
- 2(r-4) -1$ and minimum distance $d^{'}\geq r-2$. Since $m=2$
then $k^{'} - k = 2$, where $k^{'}$ denotes the dimension of
$C^{'}$ and $k$ is the dimension of $C$. We know that
$\lceil\frac{q+1}{q}d^{'} \rceil \geq r-1$. Thus, applying
Steane's code construction one has an ${[[n, n-4(r-3), d\geq
r-1]]}_{q}$ quantum code, as required.
\end{proof}

Recall that an ${[[n, k, d]]}_{q}$ code $C$ satisfies the quantum
Singleton bound given by $k + 2d \leq n + 2$. If $C$ attains the
quantum Singleton bound, i. e.,  $k + 2d = n + 2$, then it is
called a quantum maximum distance separable (MDS) code. In the
following two examples we construct quantum MDS-BCH codes:


\subsection{Construction IV - Hermitian dual-containing BCH Codes}\label{secV}

In this subsection we present the fourth proposed construction,
which is based on finding good Hermitian dual-containing BCH
codes. Let us recall some useful concepts.

Suppose that $C$ is a linear code of length $n$ over ${\mathbb
F}_{q^{2}}$. Then its Hermitian dual code is defined by
${C}^{{\perp}_{H}} = \{ y \in \displaystyle {\mathbb
F}_{q^{2}}^{n} \mid y^{q}\cdot x = 0 \ for \ all \ x \in C\}$,
where $y^q = (\displaystyle y_{1}^{q}, \ldots, y_{n}^{q} )$
denotes the conjugate of the vector $y = (y_1, \ldots, y_n )$.

\begin{lemma}\cite[Lemma 13]{Salah:2007}\label{AA}
Assume that $\gcd(q, n)=1$. A cyclic code of length $n$ over
${\mathbb F}_{q^{2}}$ with defining set $Z$ contains its Hermitian
dual code if and only if $Z\cap Z^{-q} =\emptyset$, where
$Z^{-q}=\{-qz \ mod \ n\mid z \in Z\}$.
\end{lemma}

\begin{lemma}\label{Hermite}\cite[Lemma 17 c)]{Salah:2007}
(Hermitian Construction) If there exists a classical linear ${[n,
k, d]}_{q^{2}}$ code $D$ such that ${D}^{{\perp}_{h}} \subset D$,
then there exists an ${[[n, 2k - n, \geq d]]}_{q}$ stabilizer
code.
\end{lemma}

\begin{example}
Let us start with an example of how Lemma~\ref{fewco} can be
applied together the Hermitian construction in order to construct
good codes. Assume that $q = 7$, $n = 144$, $m=3$ and $r=3$; the
$q^2$-ary cosets ${\mathbb{C}}_{3}$, ${\mathbb{C}}_{6}$,
${\mathbb{C}}_{9}$ and ${\mathbb{C}}_{12}$ contain only one
element. The other cosets necessary for the construction are
${\mathbb{C}}_{4}=\{4, 52, 100\}$, ${\mathbb{C}}_{5}=\{5, 101,
53\}$, ${\mathbb{C}}_{7}=\{7, 55, 103\}$, ${\mathbb{C}}_{8}= \{ 8,
104, 56\}$, ${\mathbb{C}}_{10}= \{ 10, 58, 106\}$,
${\mathbb{C}}_{11}=\{11, 107, 59\}$. Let $C$ be the cyclic code
generated by the product of the minimal polynomials
$M^{(3)}(x)M^{(4)}(x)M^{(5)}(x)M^{(6)}(x)M^{(7)}(x)
M^{(8)}(x)M^{(9)}(x)\cdot$ $\cdot
M^{(10)}(x)M^{(11)}(x)M^{(12)}(x)$. It is straightforward to show
that $C$ is Hermitian dual-containing and has parameters ${[144,
122, d\geq 11]}_{7^2}$. Thus, applying the Hermitian construction,
one obtains an ${[[144, 100, d\geq 11]]}_{7}$ quantum code.
Similarly one can construct quantum codes with parameters ${[[144,
102, d\geq 10]]}_{7}$, ${[[144, 108, d\geq 9]]}_{7}$, ${[[144,
114, d\geq 8]]}_{7}$, ${[[144, 116, d\geq 7]]}_{7}$, ${[[144, 122,
d\geq 6]]}_{7}$, ${[[144, 128, d\geq 5]]}_{7}$, ${[[144, 130,
d\geq 4]]}_{7}$ and ${[[144, 136, d\geq 3]]}_{7}$.
\end{example}

\begin{theorem}\label{nonpri3A}
Suppose that $q > 3$ is a prime power and $n > q^2$ is an integer
such that $\gcd (q^2, n) = 1$. Assume also that $(q^2 - 1) \mid n$
and $m= \ {{\operatorname{ord}}_{n}}(q^2)=2$ hold. Then there
exist quantum codes with parameters ${[[n, n-4(r-2)-2, d\geq
r]]}_{q}$, where $r$ is such that $n = r(q^2 - 1)$.
\end{theorem}
\begin{proof}
Let $C$ be the cyclic code generated by the product of the minimal polynomials
\begin{eqnarray*}
M^{(r)}(x)M^{(r+1)}(x)\cdot\ldots \cdot M^{(2r-2)}(x).
\end{eqnarray*}
We first show that $C$ is Hermitian dual-containing. For this, let
us consider the defining set $Z$ of $C$ consisting of the
$q^2$-ary cyclotomic cosets given by ${\mathbb{C}}_{[r]}= \{ r\},
{\mathbb{C}}_{[r+1]}= \{ r+1, \ \ r + q^2\}, {\mathbb{C}}_{[r+2]}=
\{ r+2, \ \ r + 2q^2\}, \ldots, {\mathbb{C}}_{[2r-2]}= \{ 2r-2, \
\ r + (r-2)q^2 \}.$

We know that $\gcd (q, n) = 1$ holds. From Lemma~\ref{AA}, it
suffices to show that $Z\cap Z^{-q}=\emptyset$. Seeking a
contradiction, we assume that $Z\cap Z^{-q}\neq\emptyset$. Then
there exist $i, j$, where $0 \leq i, j\leq r-2$, such that
$(r+j)q^l\equiv -q(r+i)$ mod $n$, where $l=0$ or $l=2$. If $l=0$,
one has $r+j\equiv -q(r+i)$ mod $n$ and so $q(r+i)+ r + j \equiv
0$ mod $n$. Since $q(r+i)+ r + j < n$ and $q(r+i)+ r + j \neq 0$
hold, one has a contradiction. If $l=2$, it implies that
$(r+j)q^2\equiv -q(r+i)$ mod $n$ and since $\gcd (q^2, n) = 1$ and
$rq^2 \equiv r$ mod $n$ one obtains
\begin{eqnarray*}
(r+j)q^2\equiv -q(r+i) \mod n\\
\Longrightarrow r+jq^2\equiv -q(r+i) \mod n\\
\Longrightarrow (q+1)r\equiv -q(i + jq) \mod n\\
\Longrightarrow -q(i + jq)(q-1)\equiv 0 \mod n\\
\Longrightarrow n \mid q(i + jq)(q-1)\\
\Longrightarrow r(q+1)\mid q(i + jq).
\end{eqnarray*}
Since $\gcd (r, q) = 1$ and $\gcd (q+1, q) = 1$ hold it implies
that $r(q+1)\mid (i + jq)$, which is a contradiction because $i +
jq < r(q+1)$. Thus $C$ is Hermitian dual-containing.

It is easy to see that these cosets are mutually disjoint, with
exception of the coset ${\mathbb{C}}_{[r]}$, the other cosets have
two elements. Thus $C$ has dimension $k=n-2(r-2)-1.$ By
construction, the defining set $Z$ of $C$ contains the sequence
$r, r+1, \ldots, 2r-2$, of $r-1$ consecutive integers and, so the
minimum distance of $C$ is greater than or equal to $r$, that is,
$C$ is an ${[n, n-2(r-2)-1, d\geq r]}_{q^2}$ code. Applying the
Hermitian construction to $C$ one can get an ${[[n, n-4(r-2)-2,
d\geq r]]}_{q}$ quantum code, as desired.
\end{proof}

\begin{corollary}\label{nonpri3B}
Suppose $q > 3$ is a prime power and $n > q^2$ is an integer such
that $\gcd (q^2, n) = 1$. Assume also $(q^2 - 1) \mid n$ and $m= \
{{\operatorname{ord}}_{n}}(q^2)=2$. Then there exist quantum codes
with parameters ${[[n, n-4c-2, d\geq c+2]]}_{q}$, where $2\leq c <
r-2$ and $n = r(q^2 - 1)$.
\end{corollary}
\begin{proof}
Let $C$ be the BCH code generated by the product of the minimal
polynomials $M^{(r)}(x)M^{(r+1)}(x)\cdot\ldots \cdot
M^{(r+c)}(x)$. Proceeding similarly as in the proof of
Theorem~\ref{nonpri3A}, the result follows.
\end{proof}

\begin{theorem}\label{nonpri7h}
Let $q \geq 3$ be a prime power, $n > q^2$ be a prime number and
consider that $m= \ {{\operatorname{ord}}_{n}}(q^2)\geq 2$. Let
${\mathbb{C}}_{[s]}$ be the cyclotomic coset containing $s$ and
$s+1$. Assume that $Z={\mathbb{C}}_{[s]}\cup
{\mathbb{C}}_{[s+2]}\cup \ldots \cup {\mathbb{C}}_{[s+r]}$, where
all the $q$-ary cosets ${\mathbb{C}}_{[s+i]}$, $i =0, 2, 3,\ldots,
r$, are mutually disjoint, and suppose that $Z\cap
Z^{-q}=\emptyset$. Then there exist quantum codes with parameters
${[[n, n-2mr, d\geq r+2]]}_{q}$.
\end{theorem}

\begin{proof}
We know that $\gcd (q, n) = 1$ holds. Let $C$ be the cyclic code
generated by the product of the minimal polynomials
\begin{eqnarray*}
M^{(s)}(x)M^{(s+2)}(x)\cdot\ldots\cdot M^{(s+r)}(x).
\end{eqnarray*}
Since $Z\cap Z^{-q}=\emptyset$ holds, it follows from
Lemma~\ref{AA} that $C$ is Hermitian dual-containing. From the BCH
bound, the minimum distance of $C$ is greater than or equal to
$r+2$. It is easy to see that the cosets ${\mathbb{C}}_{[s+i]}$,
where $i=0, 2, 3, \ldots, r$, have $m$ elements and they are
mutually disjoint. Thus $C$ has parameters ${[n, n - mr, d\geq
r+2]}_{q^2}$. Applying the Hermitian construction one can get an
${[[n, n - 2mr, d\geq r+2]]}_{q}$ code.
\end{proof}

We finish this subsection by showing how Lemma~\ref{nonpril1}
works for constructing quantum MDS-BCH codes:

\begin{example}
Let us consider $q=5$ and $n = 13$. Since $\gcd (13, 24) =1$, the
linear congruence $(q^2 -1)x\equiv 1$ mod $n$ has a solution, so
there exists at least one $q^2$-ary coset containing two
consecutive integers, namely, the coset ${\mathbb{C}}_{[6]}= \{6,
7\}$. Choose $C = \langle M^{(6)}(x)\rangle$. Since
${\mathbb{C}}_{[4]} \neq {\mathbb{C}}_{[6]}$, $C$ is Hermitian
dual-containing and has parameters ${[13, 11, d\geq 3]}_{5}$.
Applying the Hermitian construction, an ${[[13, 9, 3]]}_{5}$
quantum MDS-BCH code is constructed. Similarly, we can also
construct an ${[[17, 13, 3]]}_{4}$ and an ${[[17, 9, 5]]}_{4}$
quantum MDS-BCH code.
\end{example}

\section{Code Comparisons}\label{sec4}

In this section we compare the parameters of the new quantum BCH
codes with the ones available in the literature. The codes
available in the literature derived from Steane's code
construction are generated by the same method presented in
\cite[Table I]{Steane:1999} by considering the criterion for
classical Euclidean dual-containing BCH codes given in
\cite[Theorems~3 and 5]{Salah:2007}.

Let us fix the notation:
\begin{itemize}
\item ${[[n, k, d]]}_{q}$ are the parameters of the new quantum codes;

\item ${[[n^{'}, k^{'}, d^{'}]]}_{q}= \\ {[[n^{'},
n^{'}-2m(\lceil(\delta-1)(1-1/q)\rceil), d^{'}\geq\delta]]}_{q}$
are the parameters of quantum codes available in
\cite{Salah:2007};

\item ${[[n^{''}, k^{''}, d^{''}]]}_{q}$ are the parameters of
quantum BCH codes derived from Steane's code construction shown in
\cite[Corollary 4]{Hamada:2008}.
\end{itemize}

Tables~\ref{table1} and \ref{table2} show the new codes derived
from Construction~I and from Theorem~\ref{nonpri7} in Construction
II; Table~\ref{table3} presents new codes derived from
Construction~III and Table~\ref{table4} shows the new codes
derived from Construction~IV.

Checking the parameters of the new quantum BCH codes tabulated,
one can see that the new codes have parameters better than the
ones available in the literature. In other words, fixing $n$ and
$d$, the new quantum BCH codes achieve greater values of the
number of qudits than the quantum BCH codes available in the
literature. As the referee observed, it is interesting to note
that most of our codes of length larger than $q^2 + 1$ are new.

\begin{table}[!pt]
\begin{center}
\caption{Code Comparison \label{table1}}
\begin{tabular}{|c |c |}
\hline  New CSS codes & CSS codes in \cite{Salah:2007}\\
\hline ${[[n, k, d]]}_{q}$ & ${[[n^{'}, k^{'}, d^{'}]]}_{q}$\\
\hline
\hline ${[[11, 1, d\geq 4]]}_{3}$ & ---\\
\hline ${[[13, 1, d \geq 4]]}_{3}$ & ---\\
\hline ${[[1093, 1079, d \geq 3]]}_{3}$ & ${[[1093, 1065, d^{'} \geq 3]]}_{3}$\\
\hline ${[[31, 19, d\geq 4]]}_{5}$ & ${[[31, 13, d^{'}\geq 4]]}_{5}$\\
\hline ${[[31, 13, d\geq 5]]}_{5}$ & ${[[31, 7, d^{'}\geq 5]]}_{5}$\\
\hline ${[[71, 61, d \geq 3]]}_{5}$ & ${[[71, 51, d^{'} \geq 3]]}_{5}$\\
\hline ${[[71, 51, d\geq 4]]}_{5}$ & ${[[71, 41, d^{'} \geq 4]]}_{5}$\\
\hline ${[[73, 61, d\geq 4]]}_{8}$ & ${[[73, 55, d^{'}\geq 4]]}_{8}$\\
\hline ${[[73, 55, d\geq 5]]}_{8}$ & ${[[73, 49, d^{'}\geq 5]]}_{8}$\\
\hline ${[[73, 49, d\geq 6]]}_{8}$ & ${[[73, 43, d^{'}\geq 6]]}_{8}$\\
\hline ${[[73, 43, d\geq 7]]}_{8}$ & ${[[73, 37, d^{'}\geq 7]]}_{8}$\\
\hline
\end{tabular}
\end{center}
\end{table}

\begin{table}[!pt]
\begin{center}
\caption{Code Comparison \label{table2}}
\begin{tabular}{|c |c |}
\hline New CSS codes & Steane's code construction\\
\hline ${[[n, k, d]]}_{q}$ & ${[[n^{''}, k^{''}, d^{''}]]}_{q}$: $L, L^{'}$\\
\hline
\hline ${[[31, 19, d \geq 4]]}_{5}$ & ${[[31, 16, d^{''}\geq 4]]}_{5}$: ${[31, 22, 4]}_{5}$, ${[31, 25, 3]}_{5}$\\
\hline  ${[[31, 13, d \geq 5]]}_{5}$ & ${[[31, 10, d^{''}\geq 5]]}_{5}$: ${[31, 19, 5]}_{5}$, ${[31, 22, 4]}_{5}$\\
\hline  ${[[73, 61, d \geq 4]]}_{8}$ & ${[[73, 58, d^{''}\geq 4]]}_{8}$: ${[73, 64, 4]}_{8}$, ${[73, 67, 3]}_{8}$\\
\hline  ${[[73, 55, d \geq 5]]}_{8}$ & ${[[73, 52, d^{''}\geq 5]]}_{8}$: ${[73, 61, 5]}_{8}$, ${[73, 64, 4]}_{8}$\\
\hline  ${[[73, 49, d \geq 6]]}_{8}$ & ${[[73, 46, d^{''}\geq 6]]}_{8}$: ${[73, 58, 6]}_{8}$, ${[73, 61, 5]}_{8}$\\
\hline  ${[[73, 43, d \geq 7]]}_{8}$ & ${[[73, 40, d^{''}\geq 7]]}_{8}$: ${[73, 55, 7]}_{8}$, ${[73, 58, 6]}_{8}$\\
\hline
\end{tabular}
\end{center}
\end{table}

\begin{table}[!pt]
\begin{center}
\caption{Code Comparison \label{table3}}
\begin{tabular}{|c |c |}
\hline New codes (Construction III) & Steane's code construction\\
\hline ${[[n, k, d]]}_{q}$ & ${[[n^{''}, k^{''}, d^{''}]]}_{q}$\\
\hline
\hline ${[[31, 22, d \geq 4]]}_{5}$ & ${[[31, 16, d^{''}\geq 4]]}_{5}$\\
\hline ${[[31, 16, d \geq 5]]}_{5}$ & ${[[31, 10, d^{''}\geq 5]]}_{5}$\\
\hline ${[[71, 56, d \geq 4]]}_{5}$ & ${[[71, 46, d^{''}\geq 4]]}_{5}$\\
\hline ${[[73, 64, d \geq 4]]}_{8}$ & ${[[73, 58, d^{''}\geq 4]]}_{8}$\\
\hline ${[[73, 58, d \geq 5]]}_{8}$ & ${[[73, 52, d^{''}\geq 5]]}_{8}$\\
\hline ${[[40, 36, 3]]}_{9}$ (MDS) &\\
\hline ${[[60, 56, 3]]}_{11}$ (MDS) &\\

\hline
\end{tabular}
\end{center}
\end{table}

\begin{table}[!pt]
\begin{center}
\caption{Code Comparison \label{table4}}
\begin{tabular}{|c |c |}
\hline New Hermitian Codes (Construction IV)&  Hermitian Codes in \cite{Salah:2007}\\
\hline ${[[n, k, d]]}_{q}$ & ${[[n^{'}, k^{'}, d^{'}]]}_{q}$\\
\hline
\hline ${[[17, 13, 3]]}_{4}$ (MDS) &\\
\hline ${[[17, 9, 5]]}_{4}$ (MDS) &\\
\hline ${[[13, 9, 3]]}_{5}$ (MDS) & \\
\hline ${[[312, 298, d\geq 5]]}_{5}$ & ${[[312, 296, d^{'}\geq 5]]}_{5}$\\
\hline ${[[312, 294, d\geq 6]]}_{5}$ & ${[[312, 292, d^{'}\geq 6]]}_{5}$\\
\hline ${[[312, 290, d\geq 7]]}_{5}$ & ${[[312, 288, d^{'}\geq 7]]}_{5}$\\
\hline ${[[312, 286, d\geq 8]]}_{5}$ & ${[[312, 284, d^{'}\geq 8]]}_{5}$\\
\hline ${[[312, 282, d\geq 9]]}_{5}$ & ${[[312, 280, d^{'}\geq 9]]}_{5}$\\
\hline ${[[312, 278, d\geq 10]]}_{5}$ & ${[[312, 276, d^{'}\geq 10]]}_{5}$\\
\hline ${[[312, 274, d\geq 11]]}_{5}$ & ${[[312, 272, d^{'}\geq 11]]}_{5}$\\
\hline ${[[312, 270, d\geq 12]]}_{5}$ & ${[[312, 268, d^{'}\geq 12]]}_{5}$\\
\hline ${[[144, 128, d\geq 5]]}_{7}$ & ${[[144, 120, d\geq 5]]}_{7}$\\
\hline ${[[144, 122, d\geq 6]]}_{7}$ & ${[[144, 114, d\geq 6]]}_{7}$\\
\hline ${[[144, 116, d\geq 7]]}_{7}$ & ${[[144, 108, d\geq 7]]}_{7}$\\
\hline ${[[144, 114, d\geq 8]]}_{7}$ & ${[[144, 102, d\geq 8]]}_{7}$\\
\hline ${[[144, 108, d\geq 9]]}_{7}$ & ${[[144, 96, d\geq 9]]}_{7}$\\
\hline ${[[144, 102, d\geq 10]]}_{7}$ & ${[[144, 90, d\geq 10]]}_{7}$\\
\hline ${[[144, 100, d\geq 11]]}_{7}$ & ${[[144, 84, d\geq 11]]}_{7}$\\
\hline
\end{tabular}
\end{center}
\end{table}

\begin{remark}
Note that the codes ${[[31, 25, d\geq 3]]}_{5}$ and ${[[1093,
1079, d \geq 3]]}_{3}$ have the same parameters of the
corresponding Hamming codes and the new ${[[71, 61, d \geq
3]]}_{5}$ code can be compared with distance three codes obtained
by shortening Hamming codes.
\end{remark}

\section{Summary}\label{sec5}

We have presented four quantum code constructions generating new
families of good nonprimitive non-narrow-sense quantum BCH codes.
These new quantum codes have parameters better than the ones
available in the literature. Additionally, most of these codes are
generated algebraically.

\section*{Acknowledgment}
I would like to thank the Associate Editor Patrick Hayden and the
anonymous referee for their valuable comments and suggestions that
improve significantly the quality and the presentation of this
paper. I also would like to thank Prof. Reginaldo Palazzo Jr. for
useful discussions with respect to the first quantum code
construction and Dr. J. H. Kleinschmidt for critical reading of
the manuscript. Part of this work was presented in ISITA 2012,
Honolulu-Hawaii. This research has been partially supported by the
Brazilian agencies CAPES and CNPq.

\textbf{Giuliano G. La Guardia received the M.S. degree in pure
mathematics in 1998 and the Ph.D. degree in electrical engineering
in 2008, both from the State University of Campinas (UNICAMP),
Brazil. Since 1999, he has been with the Department of Mathematics
and Statistics, State University of Ponta Grossa, where he is an
Associate Professor. His research areas include theory of
classical and quantum codes, matroid theory, and error analysis.}


\begin{thebibliography}{10}


\bibitem{Salah:2007}
S. A. Aly, A. Klappenecker, and P. K. Sarvepalli.
\newblock On quantum and classical BCH codes.
\newblock {\em IEEE Trans. Inform. Theory}, 53(3):1183--1188, 2007.

\bibitem{Ashikmin:2001}
A. Ashikhmin and E. Knill.
\newblock Non-binary quantum stabilizer codes.
\newblock {\em IEEE Trans. Inform. Theory}, 47(7):3065--3072, 2001.

\bibitem{Bierbrauer:2000}
J. Bierbrauer and Y. Edel.
\newblock Quantum twisted codes.
\newblock {\em J. Comb. Designs}, 8:174--188, 2000.

\bibitem{BoRCha:1960}
R. C. Bose and D. K. Ray-Chaudhuri.
\newblock On a class of error correcting binary group codes.
\newblock {\em Information and Control}, 3:68–-79, 1960.

\bibitem{BoRChaf:1960}
R. C. Bose and D. K. Ray-Chaudhuri.
\newblock Further results on error correcting binary group codes.
\newblock {\em Information and Control}, 3:279–-290, 1960.


\bibitem{Calderbank:1998}
A. R. Calderbank, E. M. Rains, P. W. Shor, and N. J. A. Sloane.
\newblock Quantum error correction via codes over $GF(4)$.
\newblock {\em IEEE Trans. Inform. Theory}, 44(4):1369--1387, 1998.


\bibitem{Chen:2005}
H. Chen, S. Ling, and C. P. Xing.
\newblock Quantum codes from concatenated algebraic geometric codes.
\newblock {\em IEEE. Trans. Inform. Theory}, 51(8):2915 -- 2920, 2005.

\bibitem{Cohen:1999}
G. D. Cohen, S. B. Encheva, and S. Litsyn.
\newblock On binary constructions of quantum codes.
\newblock {\em IEEE Trans. Inform. Theory}, 45(7):2495--2498, 1999.


\bibitem{Grassl:1999B}
M. Grassl and T. Beth.
\newblock Quantum BCH codes.
\newblock In {\em Proc. X Int. Symp. Theor. Elec. Eng.}, pp. 207--212, Magdeburg, Germany, 1999.

\bibitem{Grassl:2003}
M. Grassl, T. Beth, and M. R\"otteler.
\newblock On optimal quantum codes.
\newblock {\em Int. J. Quantum Inform.}, 2(1):757--766, 2004.


\bibitem{Hamada:2008}
M. Hamada.
\newblock Concatenated quantum codes constructible in polynomial time:
efficient decoding and error correction.
\newblock {\em IEEE Trans. Inform. Theory}, 54(12):5689--5704, 2008.

\bibitem{Hocque:1959}
A. Hocquenghem.
\newblock Codes correcteurs d'erreurs.
\newblock {\em Chiffres}, 2:147–-156, 1959.

\bibitem{Ketkar:2006}
A. Ketkar, A. Klappenecker, S. Kumar, and P. K. Sarvepalli.
\newblock Nonbinary stabilizer codes over finite fields.
\newblock {\em IEEE Trans. Inform. Theory}, 52(11):4892--4914, 2006.

\bibitem{LaGuardia:2009}
G. G. La Guardia.
\newblock Constructions of new families of nonbinary quantum codes.
\newblock {\em Phys. Rev. A}, 80(4):042331 (1--11), 2009.

\bibitem{LaGuardia:2010}
G. G. La Guardia and R. Palazzo Jr..
\newblock Constructions of new families of nonbinary CSS codes.
\newblock {\em Discrete Math.}, 310(21):2935--2945, 2010.

\bibitem{LaGuardia:2011}
G. G. La Guardia.
\newblock New quantum MDS codes.
\newblock {\em IEEE Trans. Inform. Theory}, 57(8):5551--5554, 2011.


\bibitem{MaLu:2006}
Z. Ma, X. Lu, K. Feng and D. Feng.
\newblock On non-binary quantum BCH codes.
\newblock {\em LNCS}, 3959:675--683, 2006.

\bibitem{Macwilliams:1977}
F. J. MacWilliams and N. J. A. Sloane.
\newblock {\em The Theory of Error-Correcting Codes}.
\newblock North-Holland, 1977.

\bibitem{Nielsen:2000}
M. A. Nielsen and I. L. Chuang.
\newblock {\em Quantum Computation and Quantum Information}.
\newblock Cambridge University Press, 2000.


\bibitem{Steane:1999}
A. M. Steane.
\newblock Enlargement of Calderbank-Shor-Steane quantum codes.
\newblock {\em IEEE Trans. Inform. Theory}, 45(7):2492--2495, 1999.


\bibitem{Xi:2004}
L. Xiaoyan.
\newblock Quantum cyclic and constacyclic codes.
\newblock {\em IEEE Trans. Inform. Theory}, 50(3):547--549, 2004.


\bibitem{Yue:2000}
D.-W. Yue and G.-Z Feng.
\newblock Minimal cyclotomic coset representatives and their
applications to BCH codes and Goppa codes.
\newblock {\em IEEE Trans. Inform. Theory}, 46(7):2625--2628, 2000.


\end{thebibliography}
\end{document}